\documentclass[12pt]{article}
\usepackage{amssymb, amsmath, amsthm}

\usepackage[paper=a4paper,left=25mm,right=25mm,top=25mm,bottom=25mm]{geometry}

\usepackage{bbm,mathtools,enumitem}

\usepackage[noadjust]{cite}

\usepackage{hyperref}

\usepackage{xcolor}
\hypersetup{
  colorlinks   = true, 
  urlcolor     = {blue!90!black}, 
  linkcolor    = {blue!90!black}, 
  citecolor   = {red!90!black} 
}

\allowdisplaybreaks[3]
\numberwithin{equation}{section}

\newtheorem{Theorem}{Theorem}[section]
\newtheorem{Lemma}[Theorem]{Lemma}

\newtheorem{Proposition}[Theorem]{Proposition}

\newcommand{\CB}{\mathcal B}
\newcommand{\CD}{\mathcal D}

\newcommand{\CG}{\mathcal G}

\newcommand{\CL}{\mathcal L}

\newcommand{\CR}{\mathcal R}
\newcommand{\CX}{\mathcal X}

\newcommand{\LK}{\mathbb K}
\newcommand{\LM}{\mathbb M}
\newcommand{\LN}{\mathbb N}

\newcommand{\LR}{\mathbb R}

\newcommand{\X}{X} 
\newcommand{\Y}{{{\LR^*_+}\times{\X}}} 
\newcommand{\IC}{\boldsymbol{C}} 
\DeclareMathOperator*{\esssup}{ess\,sup}
\DeclareMathOperator*{\essinf}{ess\,inf}
\let\hat=\widehat

\newcommand\bbs{B_\mathrm{bs}(\LK_0(\X))}

\newcommand\fin{\CL_\infty(\LK_0(\X))}

\newcommand\term[1]{\textbf{\textit{#1}}}

\usepackage[noblocks]{authblk}

\author[1]{Dmitri~Finkelshtein\thanks{{\tt d.l.finkelshtein@swansea.ac.uk}.}}
\affil[1]{Department of Mathematics,
Swansea University, Bay Campus, Fabian Way, Swansea SA1 8EN, U.K.}
\author[2]{Yuri~Kondratiev\thanks{{\tt kondrat@math.uni-bielefeld.de}.}} 
\author[2]{Peter Kuchling\thanks{{\tt pkuchlin@math.uni-bielefeld.de}.}} 
\affil[2]{Fakult\"{a}t f\"{u}r Mathematik, Universit\"{a}t Bielefeld, Postfach 110 131, 33501 Bielefeld, Germany}

\title{Markov dynamics on the cone of~discrete~Radon~measures}

\begin{document}

\maketitle

\section{Introduction}

Configuration spaces form an important and actively developing area in the infinite dimensional analysis.
The spaces not only contain rich mathematical structures which require non-trivial combination of continuous and combinatoric analysis, they also provide a natural mathematical framework for the applications to mathematical physics, biology, ecology, and beyond.

Spaces of discrete Radon measures (DRM) may be considered as generalizations of configuration spaces. Main peculiarity of a DRM is that its support is typically not a configuration (i.e. not a locally finite set). The latter changes drastically the techniques for the study of the spaces of DRM. 

Spaces of DRM have various motivations coming from mathematics and applications. In particular, random DRM appear in the context of the Skorokhod theorem \cite{Sk} in the theory of processes with independent 
increments. Next, in the representation theory of current groups, the role of measures on spaces of DRM was 
clarified in fundamental works by Gelfand, Graev, and Vershik; see \cite{KLV} for the development of this approach. Additionally, DRM gives a useful technical equipment in the study of several models in mathematical physics, biology, and ecology.

In the present paper, we start with a brief overview of the known facts about the spaces of DRM (Section 2). In \cite{K}, the concept of Plato subspaces of the spaces of marked configurations was introduced. Using this, one can define topological, differential and functional structures on spaces of DRM, as well as transfer the harmonic analysis considered in \cite{MR1914839} to the spaces of DRM. This allows us to extend the study of non-equilibrium dynamics, see e.g. \cite{MR2863863,MR2417815,MR2426716}, to the spaces of DRM. 

Namely, we consider three dynamics on the spaces of DRM: an analogue of the contact model (Section~3), where we generalise some results obtained in \cite{MR2426716} and provide new two-sides estimates for the correlation functions; and analogues of the Bolker--Dieckmann--Law--Pacala model (Section~4) and Glauber-type dynamics (Section~5) where we show how the results obtained previously for the configuration spaces (see e.g. \cite{MR2863863}) can be modified for the case of the spaces of DRM.  

\section{Framework}

\subsection*{Cone of discrete Radon measures} Let $\X$ be a locally compact Polish space, and let $\CB_c(\X)$ denote the family of all Borel sets from $\X$ with a compact closure. The \term{cone of nonnegative discrete Radon measures} on $\X$ is defined as follows:
\begin{equation*}
 \LK(\X):=\Bigl\{\eta=\sum_i s_i\delta_{x_i}\in\LM(\X) \Bigm\vert s_i\in (0,\infty), x_i\in\X\Bigr\}.
\end{equation*}
By convention, the zero measure $\eta=0$ is included in $\LK(\X)$. The \term{support} of an $\eta\in\LK(\X)$ is given by
\begin{equation*}
 \tau(\eta):=\bigl\{x\in\X: 0<\eta(\{x\})=:s_x(\eta)\bigr\},
\end{equation*}
and $\tau(0):=\emptyset$. 
If $\eta$ is fixed and $x\in\tau(\eta)$, we write $s_x:=s_x(\eta)$. Therefore, 
\[
\eta(\Lambda)=\sum_{x\in\tau(\eta)\cap\Lambda}s_x<\infty, \qquad \Lambda\in \CB_c(\X), \ \eta\in\LK(\X).
\]
We stress that, in general, the number of points $|\tau(\eta)|$ in the support of a measure $\eta\in\LK(\X)$ may be infnite. Let henceforce $|\cdot|$ denote the number of elements of a set.

For $\eta,\xi\in\LK(\X)$ we write $\xi\subset\eta$ if $\tau(\xi)\subset\tau(\eta)$ and $s_x(\xi)=s_x(\eta)$ for all $x\in\tau(\xi)$. If,~additionally, $|\tau(\xi)|<\infty$, we write $\xi\Subset\eta$.

We fix the \term{vague topology} on $\LM(\X)$, which is the coarsest topology such that the mappings
\begin{equation*}
 \eta\mapsto\langle f,\eta\rangle:=\sum_{x\in\tau(\eta)}s_x f(x).
\end{equation*}
are continuous for all continuous functions $f:\X\to\LR$ with compact support. We~endow $\LK(\X)$ with the corresponding subspace topology, and also let $\CB(\LK(\X))$ denote the corresponding Borel $\sigma$-algebra.

\subsection*{Configuration spaces} Let $Y$ be a locally compact Polish space. The \term{space of locally finite configurations} over $Y$ is defined as follows:
 \begin{equation*}
  \Gamma(Y)=\bigl\{\gamma\subset Y : |\gamma\cap \Delta|<\infty \text{ for all compact } \Delta\subset Y\bigr\}.
 \end{equation*}
 Then $\Gamma(Y)$ is naturally embedded into the space of Radon measures $\Gamma(Y)\subset\LM(Y)$; we endow it with the vague topology defined on $\LM(Y)$.  Let $\CB(\Gamma(Y))$ be the corresponding Borel $\sigma$-algebra. 

We denote $\LR_+^*:=(0,\infty)$ and consider $Y=\Y$. Let $\Gamma_\mathrm{p}(\Y)\subset\Gamma(\Y)$ denote the set of all \term{pinpointing configurations}; the latter means  that $\gamma\in \Gamma_\mathrm{p}(\Y)$ iff $(s_1,x),(s_2,x)\in\gamma$ for an $x\in\X$ implies $s_1=s_2$. 

For a pinpointing configuration $\gamma\in\Gamma_\mathrm{p}(Y)$, we introduce the local mass of a pre-compact set $\Lambda\in\CB_c(\X)$:
 \begin{equation*}
  \gamma(\Lambda)=\int_{\LR^*_+\times\X}s\mathbbm{1}_\Lambda(x)\ d\gamma(s,x)=\sum_{(s,x)\in\gamma}s\mathbbm{1}_\Lambda(x)\in[0,\infty].
 \end{equation*}
Finally, we define the space of pinpointing configurations with finite local mass:
\begin{equation*}
 \Pi(\Y):=\bigl\{\gamma\in\Gamma_\mathrm{p}(\Y) : \gamma(\Lambda)<\infty \text{ for all } \Lambda\in\CB_c(\X)\bigr\}.
\end{equation*}
We endow $\Pi(\Y)$ with the subspace topology coming from $\Gamma(\Y)$, and one can consider the corresponding (trace) Borel $\sigma$-algebra.

The mapping $\CR:\Pi(\Y)\to\LK(\X)$ given by
\begin{equation}\label{defR}
\gamma=\sum_{(s,x)\in\gamma}\delta_{(s,x)}\longmapsto \CR \gamma:=\sum_{(s,x)\in\gamma}s\delta_x
\end{equation}
provides a natural bijection. It can be shown that both $\CR$ and $\CR^{-1}$ are measurable with respect to the Borel $\sigma$-algebras constructed above, i.e. 
$\CB(\Pi(\Y))$ and $\CB(\LK(\X))$ are $\sigma$-isomorphic, see \cite{MR3041709}.

\subsection*{Discrete measures with finite support}

We consider the subcone of all discrete nonnegative Radon measures with finite support:
\begin{gather*}
 \LK_0(\X):=\bigl\{\eta\in\LK(\X): |\tau(\eta)|<\infty\bigr\}=\bigsqcup_{n=0}^\infty\LK_0^{(n)}(\X),\\
\shortintertext{where}
 \LK_0^{(n)}(\X):=\left\{\eta\in\LK_0(\X): |\tau(\eta)|=n\right\},\ n\in\LN;
 \qquad \LK_0^{(0)}(\X):=\{0\}.
\end{gather*}
The mapping $\CR$, given by \eqref{defR}, provides provides a bijection between $\LK_0(\X)$ and the set $\Gamma_{0,\mathrm{p}}(\Y)$ of pinpointing finite configurations on $\Y$. We define the Borel $\sigma$-algebra on  $\LK_0(\X)$ as the smallest $\sigma$-algebra which makes this mapping $\CR$ measurable. 

Any measurable function $G:\LK_0(\X)\to\LR$ can be identified with the sequence of symmetric functions on $(\Y)^n$, $n\in\LN$, through the equalities:
 \begin{equation*}
 G^{(n)}(s_1,x_1,\dotsc,s_n,x_n):=G\Bigl(\sum_{i=1}^n s_i\delta_{x_i}\Bigr), \qquad
 \sum_{i=1}^n s_i\delta_{x_i}\in \LK_0^{(n)}(\X), \ n\in\LN.
\end{equation*}
We set also $G^{(0)}:=G(0)\in\LR$.
 
A set $A\subset\LK_0(\X)$ is called \term{bounded} if there exist $\Lambda\in\CB_c(\X)$, $N\in\LN$, and a segment $I:=[a,b]\subset\LR_+^*$ such that, for all $\eta\in A$,
\[
\tau(\eta)\subset\Lambda, \qquad |\tau(\eta)|\leq N, \qquad
s_x \in I \text{ for all } x\in\tau(\eta).
\]
The family of all bounded measurable subsets of $\LK_0(\X)$ is denoted by $\CB_\mathrm{b}(\LK_0(\X))$. A measure $\rho$ on $\LK_0(\X)$ is called \term{locally finite} if $\rho(A)<\infty$ for all $A\in \CB_\mathrm{b}(\LK_0(\X))$. 

An example of a locally finite measure on $\LK_0(\X)$ is the \term{Lebesgue--Poisson measure} $\lambda_{\nu\otimes\sigma}$ with the intensity measure $\nu\otimes\sigma$, where $\nu$ and $\sigma$ are non-atomic Radon measures on  $\LR^*_+$ and $\X$, respectively, and  $\nu$ has a finite first moment:
\begin{equation*}
  \int_{\LR_+^*}s\nu(ds)<\infty.
 \end{equation*}
The Lebesgue--Poisson measure $\lambda_{\nu\otimes\sigma}$ is then defined through the equality
 \begin{align*}
  \int_{\LK_0(\X)}&G(\eta)\lambda_{\nu\otimes\sigma}(d\eta)=
  \\
  &=G(0)+\sum_{n=1}^\infty\frac{1}{n!}\int_{(\LR^*_+\times\LR^d)^n}G^{(n)}(s_1,x_1,\dotsc,s_n,x_n)\nu(ds_1)\dotso\nu(ds_n)\sigma(dx_1)\dotso\sigma(dx_n),
 \end{align*}
 which should hold for any $G:\LK_0(\X)\to\LR_+$. 

We also consider a special case of the measure $\nu=\nu_\theta$, where
\begin{equation}\label{gamma}
 \nu_\theta(ds)=\frac{\theta}{s}e^{-s}ds
\end{equation}
for some $\theta>0$. For a fixed non-atomic Radon measure $\sigma$ on $\X$, we then denote
\[
  \lambda_\theta:=\lambda_{\nu_\theta\otimes\sigma}.
\]

A function $G:\LK_0(\X)\to\LR$ is said to be a bounded function with \term{bounded support} if $|G(\eta)|\leq C \mathbbm{1}_A(\eta)$, $\eta\in\LK_0(\X)$, for some $C>0$, $A\in \CB_\mathrm{b}(\LK_0(\X))$. The set of all bounded functions on $\LK_0(\X)$ with bounded support is denoted by $B_\mathrm{bs}(\LK_0(\X))$.
Clearly, for any locally finite measure $\rho$ on $\LK_0(\X)$,
\[
\int_{\LK_0(\X)} |G(\eta)|\rho(d\eta)<\infty, \qquad
G\in B_\mathrm{bs}(\LK_0(\X)).
\]
Note that $B_\mathrm{bs}(\LK_0(\X))$ is dense in $L^1(\LK_0(\X),\lambda_{\nu\otimes\sigma})$, where $\nu$ and $\sigma$ are as above.

We will need the following identity.
\begin{Lemma}[Minlos lemma]\label{Minlos}
 Let $\lambda_{\nu\otimes\sigma}$ be defined as the above.
 \begin{enumerate}
  \item Let $G:\LK_0(\X)\to\LR$, $H:(\LK_0(\X))^2\to\LR$. Then
  \begin{align*}
   \int_{\LK_0(\X)}&\int_{\LK_0(\X)}G(\xi_1+\xi_2)H(\xi_1,\xi_2) \lambda_{\nu\otimes\sigma} (d\xi_1)\lambda_{\nu\otimes\sigma} (d\xi_2)
   \\
   &=\int_{\LK_0(\X)}G(\eta)\sum_{\xi\subset\eta}H(\xi,\eta-\xi)\lambda_{\nu\otimes\sigma}(d\eta).
  \end{align*}
  \item Let $H:\LK_0(\X)\times\LR_+^*\times\LR^d\to\LR$. Then
  \begin{align*}
   \int_{\LK_0(\X)}&\sum_{x\in\tau(\eta)}H(\eta,s_x,x)\lambda_{\nu\otimes\sigma}(d\eta)
   \\
   &=\int_{\LK_0(\X)}\int_{\LR_+^*\times\LR^d} H(\eta+s\delta_x,s,x)\nu(ds) \sigma(dx)\lambda_{\nu\otimes\sigma}(d\eta),,
  \end{align*}
 \end{enumerate}
provided, at least one side of the equality exists.
\end{Lemma}

\subsection*{Harmonic analysis on the cone}

For any $G\in B_\mathrm{bs}(\LK_0(\X))$, we define $KG:\LK(\X)\to\LR$ by (cf. \cite{MR1914839})
\begin{equation}\label{defKtr}
 (KG)(\eta):=\sum_{\xi\Subset\eta}G(\xi).
\end{equation}

\begin{Proposition}[{see \cite{cone1,kuchling_phd}}]
For any $G\in B_\mathrm{bs}(\LK_0(\X))$, there exist $C>0$, $\Lambda\in\CB_c(\X)$, $N\in\LN$, and a segment $I=[a,b]\subset\LR_+^*$ such that, for each $\eta\in\LK(\X)$,
\begin{align*}
(KG)(\eta) &= (KG)\Bigl( \sum_{x\in\tau(\eta)\cap\Lambda}\mathbbm{1}_I(s_x)s_x\Bigr),\\
\bigl\lvert (KG)(\eta) \bigr\rvert & \leq C \bigl(1+|\tau(\eta)\cap\Lambda|\bigr)^N.
\end{align*}
\end{Proposition}

Note that \eqref{defKtr} can be also defined pointwise on a wider class of functions (see \cite{cone1,kuchling_phd} for details). In particular, for the Lebesgue-Poisson exponents
\[
  e_\lambda(f,\eta):=\prod_{y\in\tau(\eta)}f(s_y,y), \quad \eta\in\LK_0(\X), \qquad e_\lambda(f,0):=1,
\]
one has that
\[
Ke_\lambda(f,\eta)=\prod_{y\in\tau(\eta)}(1+f(s_y,y)),\qquad \eta\in\LK(\X),
\]
provided that e.g. $|f(s,y)|\leq C\,s \,\mathbbm{1}_\Lambda(y) $ for $(s,y)\in\Y$, where $C>0$, $\Lambda\in\CB_c(\X)$.
 
Note also that, for any $f\in L^1(\Y,d\nu d\sigma)$,
\begin{equation}\label{eq:intexp}
  \int_{\LK_0(\X)}e_\lambda(f,\eta)\lambda_{\nu\otimes\sigma}(d\eta)=\exp\left(\int_{\Y}f(s,x)\nu (ds)\sigma(dx)\right).
 \end{equation}

For measurable $G_1,G_2:\LK_0(\X)\to\LR$, we define their $\star$-convolution as follows:
\begin{equation*}
 (G_1\star G_2)(\eta)=\sum_{\substack{\xi_1+\xi_2+\xi_3=\eta:\\[1mm]
 \tau(\xi_i)\cap\tau(\xi_j)=\emptyset}}G_1(\xi_1+\xi_2)G_2(\xi_2+\xi_3).
\end{equation*}
Then, for any $G_1,G_2\in B_\mathrm{bs}(\LK_0(\X))$,
\begin{equation*}
 K(G_1\star G_2)=KG_1\cdot KG_2.
\end{equation*}
 
Let $\mu$ be a probability measure on the space $(\LK(\X),\CB(\LK(\X)))$ such that
\begin{equation*}
 \int_{\LK(\X)}|\eta(\Lambda)|^N\mu(d\eta)<\infty 
\end{equation*}
for any $\Lambda\in\CB_c(\X)$ and $N\in\LN$. Then $\mu$ is said to have \term{finite local moments} of all orders. The space of all such measures is denoted by $\mathcal{M}^1_\mathrm{fm}(\LK(\X))$.  In particular,
\[
  K(\bbs) \subset L^1(\LK(\X),\mu), \qquad \mu\in \mathcal{M}^1_\mathrm{fm}(\LK(\X)).
\]
The corresponding \term{correlation measure} $\rho_\mu$  on $(\LK_0(\X),\CB(\LK_0(\X)))$ is then defined by the relation
\begin{equation*}
 \rho_\mu(A):=\int_{\LK(\X)} (K\mathbbm{1}_A)(\eta) \mu(d\eta), \qquad A\in\CB_b(\LK_0(\X)).
\end{equation*}

\begin{Proposition}[{see \cite{cone1,kuchling_phd}}]\label{propK}
Let $\mu\in\mathcal{M}^1_\mathrm{fm}(\LK(\X))$. Then
\begin{enumerate}
  \item The corresponding correlation measure $\rho_\mu$ is locally finite. 
  \item For any $G\in L^1(\LK_0(\X),\rho_\mu)$, the sum in \eqref{defKtr} 
converges $\mu$-almost surely, and
  \item $KG\in L^1(\LK(\X),\mu)$ with
\begin{gather}
\int_{\LK(\X)}KG(\eta)\mu(d\eta)=\int_{\LK_0(\X)}G(\eta)\rho_\mu(d\eta),\label{eqsfa}\\
 \|KG\|_{L^1(\mu)}\leq\|G\|_{L^1(\rho_\mu)}.\notag
\end{gather}
\end{enumerate}
\end{Proposition} 

Let $\nu,\sigma$ be as above. Consider the Poisson measure $\pi_{\nu\otimes\sigma}$ on $\Gamma(\Y)$ with the intensity measure $\nu\otimes\sigma$ on $\Y$, then $\pi_{\nu\otimes\sigma}(\Pi(\Y))=1$ (see \cite{cone1,kuchling_phd} for details). Hence, we may view $\pi_{\nu\otimes\sigma}$ as a probability measure on $\Pi(\Y)$, and consider the corresponding push-forward measure on $\LK(\X)$ under the mapping $\CR$. This measure belongs to $\mathcal{M}^1_\mathrm{fm}(\LK(\X))$, and the corresponding correlation measure is just $\lambda_{\nu\otimes \sigma}$.

In the special case $\nu=\nu_\theta$ given through \eqref{gamma}, $\theta>0$, the corresponding push-forward measure on $\LK(\X)$ is called the \term{Gamma measure} $\CG_\theta$ with  the intensity $\theta>0$.

Let $\mu\in\mathcal{M}^1_\mathrm{fm}(\LK(\X))$ and $\rho_\mu$ be the corresponding correlation measure. A function $k_\mu:\LK_0(\X)\to\LR$ is called the correlation function of $\mu$ if it is the density of the correlation measure with respect to the Lebesgue-Poisson measure $\lambda_{\nu\otimes\sigma}$, i.e. if
\begin{equation*}
 \rho(d\eta)=k_\mu(\eta)\lambda_{\nu\otimes\sigma}(d\eta).
\end{equation*}
For sufficient conditions for the existence of the correlation function, see \cite{cone1,kuchling_phd}. 

\subsection*{Statistical dynamics}
We are going to describe evolutions of measures $\mu_0\mapsto\mu_t$ in the space $\mathcal{M}^1_\mathrm{fm}(\LK(\X))$ through a (formal) Markov generator $L$. We assume that $L$ is defined on a linear set $\CD\subset K(\bbs)$ such that $LF\in L^1(\LK(\X),\mu_t)$, $t\geq0$,  for all $F\in\CD$. Then we define the evolution of measures through the equality
\begin{equation}\label{eqDual}
  \frac{d}{dt}\int_{\LK(\X)} F(\eta) \mu_t(d\eta)=
  \int_{\LK(\X)} (LF)(\eta) \mu_t(d\eta),
\end{equation}
for all $t\geq0$, $F\in\CD$
(recall that, by Proposition \ref{propK}, $F\in L^1(\LK(\X),\mu_t)$ for $t\geq0$).
 
Rewriting $F=KG$, $G\in\bbs$, and defining $\hat{L}G$ through the identity
\begin{equation*}
K \hat{L}G = L KG, \qquad G\in\bbs,
\end{equation*}
one can rewrite \eqref{eqDual}, by using\eqref{eqsfa}, as follows:
\begin{equation}\label{eqDuala}
  \frac{d}{dt}\int_{\LK_0(\X)} G(\eta) \rho_{\mu_t}(d\eta)=
  \int_{\LK_0(\X)} (\hat{L}G)(\eta) \rho_{\mu_t}(d\eta)
\end{equation}
for all $t>0$, $G\in\bbs$. Here
\[
  \hat{L}G =K^{-1}LKG, \qquad G\in\bbs,
\]
where
\[
  (K^{-1}F)(\eta) := \sum_{\xi\subset\eta}(-1)^{|\tau(\eta)|-|\tau(\xi)|}F(\xi), \qquad \eta\in\LK_0(\X).
\]

We will restrict our attention to the dynamics of correlation measures which have correlation functions: $\rho_{\mu_t}(d\eta)=k_t(\eta)d\lambda_{\nu\otimes\sigma}(\eta)$, assuming $k_0$ be given. Then \eqref{eqDuala} can be rewritten as follows:
\begin{equation}\label{eqDualb}
  \frac{d}{dt}\int_{\LK_0(\X)} G(\eta) k_t(\eta)d\lambda_{\nu\otimes\sigma}(\eta)=
  \int_{\LK_0(\X)} (\hat{L}G)(\eta)k_t(\eta)d\lambda_{\nu\otimes\sigma}(\eta)
\end{equation}
for all $t>0$, $G\in\bbs$.

Let $L^\triangle$ denote the dual operator to $\hat{L}$, i.e.
\begin{equation}\label{eq:duality}
\int_{\LK_0(\X)} (\hat{L}G)(\eta)k(\eta)d\lambda_{\nu\otimes\sigma}(\eta)
  =\int_{\LK_0(\X)} G(\eta)(L^\triangle k)(\eta)d\lambda_{\nu\otimes\sigma}(\eta)
\end{equation}
for all $G,k:\LK_0(\X)\to\LR$, such that both sides of the latter equality are finite. Then one can rewrite \eqref{eqDualb} as follows
\begin{equation}\label{eqDualc}
  \frac{d}{dt}\int_{\LK_0(\X)} G(\eta) k_t(\eta)d\lambda_{\nu\otimes\sigma}(\eta)=
  \int_{\LK_0(\X)} G(\eta)(L^\triangle k_t)(\eta)d\lambda_{\nu\otimes\sigma}(\eta)
\end{equation}
for all $t>0$, $G\in\bbs$. The latter weak-type equation defines hence the evolution of the correlation functions generated by the Markov operator $L$. We can consider also its strong form:
\begin{equation}\label{qfpe}
  \frac{\partial}{\partial t}k_t(\eta)=L^\triangle k_t(\eta),\qquad t>0.
\end{equation}
considered on a suitable class of correlation functions.

\section{Contact model}
Let $\nu$ and $\sigma$ be non-atomic Radon measures on $\LR_+^*$ and $\X$, respectively. We define
\begin{align*}
 (LF)(\eta)&=\sum_{x\in\tau(\eta)}m(s_x)[F(\eta-s_x\delta_x)-F(\eta)]
 \\
 &\quad+\sum_{x\in\tau(\eta)}\int_{\Y}q(s_x,s)a(x-y)[F(\eta+s\delta_y)-F(\eta)]\nu(ds)\sigma(dy)
\end{align*}
for $F\in K(\bbs)$, cf. \cite{MR2426716}. Here $a:\X\to[0,\infty)$, $m:\LR_+^*
\to[0,\infty)$, $q:\LR_+^*\times\LR_+^*\to[0,\infty)$ are such that
\begin{equation}\label{cm:cond}
\begin{gathered}
 a(-x)=a(x), \quad x\in\X, \qquad a\in L^1(\X,d\sigma)\cap L^\infty(\X,d\sigma),\qquad m\in L^\infty(\LR_+^*,d\nu),\\
 q\in L^\infty(\LR_+^*\times\LR_+^*,d\nu\,d\nu), \qquad
 \int q(s',\cdot)\nu(ds') \in L^\infty(\LR_+^*,d\nu).
\end{gathered}\end{equation}

\begin{Proposition}
For any $G\in B_\mathrm{bs}(\LK_0(\X))$, 
$\hat{L}G:=K^{-1}LKG$ satisfies
 \begin{align}\label{hatL}
  (\hat{L}G)(\eta)&=-\sum_{x\in\tau(\eta)}m(s_x)G(\eta)
  \\
  &\quad+\sum_{x\in\tau(\eta)}\int_{\Y}q(s_x,s)a(x-y)[G(\eta-s_x\delta_x+s\delta_y)+G(\eta+s\delta_y)]\nu (ds)\sigma(dy). \notag
 \end{align}
\end{Proposition}
 \begin{proof}
Firstly, we note that, for any $G \in B_\mathrm{bs}(\LK_0(\X))$ and $F:=KG$,
\begin{align*}
F(\eta-s_x\delta_x)-F(\eta)&=-K(G(\cdot+s_x\delta_x))(\eta-s_x\delta_x)
 \\
 F(\eta+s_x\delta_x)-F(\eta)&=K(G(\cdot+s_x\delta_x))(\eta).
\end{align*} 
Then
  \begin{align*}
&\quad \sum_{x\in\tau(\eta)}s_x\left[F(\eta-s_x\delta_x)-F(\eta)\right]
=-\sum_{x\in\tau(\eta)}s_x K(G(\cdot+s_x\delta_x))(\eta-s_x\delta_x) 
  \\
  &=-\sum_{x\in\tau(\eta)}s_x \sum_{\xi\Subset\eta-s_x\delta_x}G(\xi+s_x\delta_x) 
 =-\sum_{\xi\Subset\eta}\sum_{x\in\tau(\xi)}s_xG(\xi); 
 \end{align*}
 and
 \begin{align*}
 &\quad \sum_{x\in\tau(\eta)}\int_{\Y}q(s_x,s)a(x-y)[F(\eta+s\delta_y)-F(\eta)]\nu(ds)\sigma(dy)\\&=
 \sum_{x\in\tau(\eta)}\int_{\Y}q(s_x,s)a(x-y)K(G(\cdot+s\delta_y))(\eta)\nu(ds)\sigma(dy)\\&=
 \sum_{x\in\tau(\eta)}\sum_{\xi\Subset\eta}\int_{\Y}q(s_x,s)a(x-y)G(\xi+s\delta_y)\nu(ds)\sigma(dy)\\&=
 \sum_{x\in\tau(\eta)}\sum_{\xi\Subset\eta-s_x \delta_x}\int_{\Y}q(s_x,s)a(x-y)G(\xi+s\delta_y)\nu(ds)\sigma(dy)\\&\quad +
 \sum_{x\in\tau(\eta)}\sum_{\xi\Subset\eta-s_x \delta_x}\int_{\Y}q(s_x,s)a(x-y)G(\xi+s_x \delta_x+s\delta_y)\nu(ds)\sigma(dy)
 \\&=
 \sum_{\xi\Subset\eta}\sum_{x\in\tau(\xi)}\int_{\Y}q(s_x,s)a(x-y)G(\xi-s_x \delta_x+s\delta_y)\nu(ds)\sigma(dy)\\&\quad +
 \sum_{\xi\Subset\eta }\sum_{x\in\tau(\xi)}\int_{\Y}q(s_x,s)a(x-y)G(\xi+s\delta_y)\nu(ds)\sigma(dy),
 \end{align*}
 that proves the statement.
\end{proof}

Let, for fixed $\nu$ and $\sigma$, 
\[
  \CX_n:=L^\infty\left((\Y)^n, (\nu\otimes\sigma)^{\otimes n}\right), \qquad n\in\LN.
\] 
Let $\|\cdot\|_n$ denote the norm in $\CX_n$. 

Let $\fin$ denote the set of all functions $k:\LK_0(\X)\to\LR$ such that $k^{(n)}\in\CX_n$ for each $n\in\LN$. Note that, for all $G\in\bbs$ and $k\in\fin$, 
\[
  \int_{\LK_0(\X)} |G(\eta) k(\eta)|\lambda_{\nu\otimes\sigma}(d\eta)<\infty.
\]

\begin{Proposition}
 For any $k\in\fin$, the mapping
 \begin{align*}
  (L^\triangle k)(\eta)=&-\sum_{x\in\tau(\eta)}m(s_x)k(\eta)
  \\
  &+\sum_{y\in\tau(\eta)}\int_{\Y}q(s,s_y)a(x-y)k(\eta-s_y\delta_y+s\delta_x)\nu (ds)\sigma(dx)
  \\
  &+\sum_{y\in\tau(\eta)}\sum_{x\in\tau(\eta)\setminus\{y\}}q(s_x,s_y)a(x-y)k(\eta-s_y\delta_y)
 \end{align*}
 is well-defined and, for any $G\in\bbs$, 
 \[
\int_{\LK_0(\X)} G(\eta) (L^\triangle k)(\eta)\lambda_{\nu\otimes\sigma}(d\eta)= \int_{\LK_0(\X)} (\hat{L}G)(\eta) k(\eta)\lambda_{\nu\otimes\sigma}(d\eta).
 \] 
\end{Proposition}
\begin{proof}
The result is a straightforward application of the Minlos lemma.
\end{proof}

\begin{Theorem}\label{apriori1}
Let \eqref{cm:cond} hold. We define, for $s>0$,
\[
  \kappa(s):=\int_{\X}a(x)\sigma(dx)\cdot\int_{\LR^*_+}q(s',s)\nu (ds'), \quad\qquad
  r(s):=\kappa(s)-m(s),
\]
and set
\[
R:=\esssup_{s>0}r(s)\in\LR.
\]
Let $0\leq k_0\in\fin$.
\begin{enumerate}
  \item There exists a unique point-wise solution to the initial value problem \eqref{qfpe}; moreover, $0\leq k_t\in\fin$.
  \item Suppose that, for some $C>0$,
  \[
  \|k_0^{(n)}\|_n \leq C^n n!, \qquad n\in\LN.
 \]
 Then, for all $ t>0$, $n\in\LN$
 \begin{equation*}
 \|k_t^{(n)}\|_n\leq \begin{cases}
 e^{tR} (C+t)^n n! & \quad \text{if } R<0,\\[2mm]
 e^{tnR} (C+t)^n n! & \quad \text{if } R\geq 0.
 \end{cases}
 \end{equation*}
 \item Denote $\mu=\|m\|_{L^\infty(\LR_+^*,d\nu)}$. Suppose that there exists  $B\subset\Y$ such that
 \begin{equation*}
  \begin{aligned}
  \alpha:&=\min\Bigl\{\essinf_{(s,x)\in B}k_0^{(1)}(s,x) , \ 
  \essinf_{(s_1,x_1),(s_2,x_2)\in B}q(s_1,s_2)a(x_1-x_2)\Bigr\}>0;\\
   \beta:&= \alpha\cdot (\nu\otimes\sigma)(B) < \mu.
  \end{aligned}
 \end{equation*}
Denote also $T_n:=\sum\limits_{j=1}^{n-1} \frac{1}{j}$ for $n\geq2$; $T_1:=0$; $\hat{x}^{(n)}:=(s_1,x_1,\ldots,s_n,x_n)$. Then 
\begin{equation*}
 k_t^{(n)}(\hat{x}^{(n)})\geq \alpha^n e^{(\beta-\mu)nt}n! \qquad
 \text{for } \hat{x}^{(n)}\in B^n, \ t\geq T_n.
 \end{equation*}

 \end{enumerate}

\end{Theorem}
\begin{proof}

1) Consider a convolution-type operator on $\CX_n$, $n\in\LN$: for $1\leq i\leq n$,
\[
 (A_i k^{(n)})(s_1,x_1,\ldots,s_n,x_n):=
\int_{\Y}q(s,s_i)a(x-x_i) k_i^{(n)}(s,x)\nu (ds)\sigma(dx),
\]
where
\begin{equation}\label{ki}
 k_i^{(n)}(s,x):= k^{(n)}(s_1,x_1,\ldots,s_{i-1},x_{i-1},s,x,s_{i+1},x_{i+1},\ldots,s_n,x_n).
\end{equation}

We define, for $k^{(n)}\in\CX_n$, $n\in\LN$,  $1\leq i\leq n$, $\hat{x}^{(n)}:=(s_1,x_1,\ldots,s_n,x_n)$
\begin{gather*}
(B_i k^{(n)})(\hat{x}^{(n)}):=m(s_i)k^{(n)}
(\hat{x}^{(n)}),\qquad
(C_i k^{(n)})(\hat{x}^{(n)}):=\kappa(s_i)k^{(n)}
(\hat{x}^{(n)}),\\
(V_i k^{(n)})(\hat{x}^{(n)}):=r(s_i)k^{(n)}
(\hat{x}^{(n)}).
\end{gather*}
Then $M_i:=A_i-C_i$ 
is the jump generator w.r.t. $i$-th variable: i.e. for fixed $s_j,x_j$, $1\leq j\leq n$, $j\neq i$,
\[
  (M_ik_i^{(n)})(s_i,x_i)=  \int_{\Y}q(s,s_i)a(x-x_i) 
  \bigl(k_i^{(n)}(s,x) -k_i^{(n)}(s_i,x_i)\bigr)\nu (ds)\sigma(dx),
\]
where $k_i^{(n)}$ is given through \eqref{ki}.

We set also
\[
A^{(n)}:= \sum_{i=1}^n A_i; \quad
B^{(n)}:= \sum_{i=1}^n B_i; \quad
V^{(n)}:= \sum_{i=1}^n V_i; \quad
M^{(n)}:= \sum_{i=1}^n M_i.
\]
Finally, we consider mappings from $\CX_{n-1}$ to $\CX_n$, $n\geq2$:
\begin{multline*}
 (W_i  k^{(n-1)})(s_1,x_1,\ldots,s_n,x_n):= 
 \Bigl(\sum_{j\neq i}q(s_j,s_i)a(x_j-x_i) \Bigr) \\\times k^{(n-1)}(s_1,x_1,\ldots,s_{i-1},x_{i-1},s_{i+1},x_{i+1},\ldots,s_n,x_n)
\end{multline*}
for $1\leq i\leq n$, and 
\[
   W^{(n)}:=\sum_{i=1}^n W_i, \qquad n\geq2.
\]
We set also $W_1:=W^{(1)}:=0$.

It is straightforward to see that, under assumptions \eqref{cm:cond}, operators $A_i,B_i,V_i,M_i$, and hence $A^{(n)},B^{(n)},V^{(n)},M^{(n)}$, are bounded linear operators on $\CX_n$; and also $W_i$ and $W^{(n)}$ are linear bounded operators from $\CX_{n-1}$ to $\CX_n$.

The initial value problem \eqref{qfpe} can be hence rewritten as follows
\begin{align*}
\frac{\partial}{\partial t} k^{(n)}_t &= A^{(n)}k^{(n)}_t -B^{(n)}k^{(n)}_t +W^{(n)}k^{(n-1)}
\\&=  M^{(n)}k_t^{(n)} + V^{(n)}k_t^{(n)} +W^{(n)}k_t^{(n-1)};\qquad\qquad n\in\LN\\
k^{(n)}_0&\in\CX_n.
\end{align*}
Since $W^{(1)}k_t^{(0)}=0$ and all operators are bounded, the latter system can be solved recursively:
\begin{equation}\label{eq:repr}
\begin{aligned}
 k_t^{(n)}(s_1,x_1,\dotsc,s_n,x_n)&=e^{t(M^{(n)}+V^{(n)})}k_0^{(n)}(s_1,x_1,\dotsc,s_n,x_n)
 \\
 &\quad+\int_0^t e^{(t-\tau)(M^{(n)}+V^{(n)})}(W^{(n)}k_\tau^{(n-1)})(s_1,x_1,\dotsc,s_n,x_n)d\tau
\end{aligned}
\end{equation}

Let $\CX_n^+$ denote the cone of all non-negative (a.e.) functions in $\CX_n$, $n\in\LN$. By the Trotter--Lie formula for bounded operators, 
\[
  e^{t M^{(n)}}=\lim_{m\to\infty}\bigl(e^{\frac{t}{m}A^{(n)}}e^{-\frac{t}{m}C^{(n)}}\bigr)^m.
\]
By the very definition, $A^{(n)}:\CX_n^+\to\CX_n^+$, hence,
\[
  e^{t A^{(n)}}=\sum_{j=0}^\infty\frac{t}{j!}(A^{(n)})^j:\CX_n^+\to\CX_n^+, \qquad t\geq0.
\]
Next, $e^{-tC^{(n)}}$ is just a multiplication operator by a non-negative function, hence, it preserves $\CX_n^+$ as well. As a result, $e^{t M^{(n)}}:\CX_n^+\to\CX_n^+$. Using again the Trotter--Lie formula for 
\begin{equation}\label{eqasffas}
e^{t(M^{(n)}+V^{(n)})}=\lim_{m\to\infty}\bigl(e^{\frac{t}{m}M^{(n)}}e^{\frac{t}{m}V^{(n)}}\bigr)^m,
\end{equation}
we conclude by the same arguments that it also preserves $\CX_n^+$. Since $W^{(n)}:\CX_{n-1}^+\to\CX_n^+$, we get recursively from \eqref{eq:repr} that $k_0^{(n)}\in\CX_n^+$, $n\in\LN$, implies $k_t^{(n)}\in\CX_n^+$, $n\in\LN$, $t>0$.

2) Since $M^{(n)}1=0$, we have that $e^{tM^{(n)}}1=1$. Since $e^{t M^{(n)}}$ preserves $\CX_n^+$, we have, for any $f_n\in\CX_n^+$, which hence satisfies the inequality $0\leq f_n\leq \|f_n\|_n$, that $0\leq e^{tM^{(n)}} f_n\leq e^{tM^{(n)}}\|f_n\|_n=\|f_n\|_n$, and thus
\[
  \|e^{tM^{(n)}}f_n\|_n\leq \|f_n\|_n, \qquad f_n\in \CX_n^+.
\]
Since $e^{tV^{(n)}}$ is a multiplication operator, 
\[
  \|e^{t V^{(n)}}\|= e^{t \esssup V^{(n)}(\hat{x}^{(n)})}\leq e^{tnR}.
\]
   Therefore,  by \eqref{eqasffas},
\[
  \|e^{t(M^{(n)}+V^{(n)})}f_n\|_n\leq e^{tnR}\|f_n\|_n, \qquad f_n\in\CX_n^+.
\]
Then, by \eqref{eq:repr},
\begin{align*}
 \|k_t^{(n)}\|_n&\leq e^{tnR} \|k_0^{(n)}\|_n
+\int_0^t e^{(t-\tau)nR}\|W^{(n)}k_\tau^{(n-1)}\|_n d\tau\\
& \leq e^{tnR} \|k_0^{(n)}\|_n
+n(n-1)\int_0^t e^{(t-\tau)nR}\|k_\tau^{(n-1)}\|_{n-1} d\tau.
\end{align*}
 For $n=1$, it reads as
 \[
 \|k_t^{(1)}\|_1\leq e^{tR}\|k_0^{(1)}\|_1\leq Ce^{tR}\leq (C+t)^ne^{tR}.
 \]

 For $n\geq2$, consider two cases separately.

 Let $R<0$. Then, assuming that 
 \[
 \|k_\tau^{(n-1)}\|_{n-1}\leq e^{\tau R} (C+\tau)^{n-1}(n-1)!, \qquad \tau\geq0, 
 \]
 and using the inequality $e^{(t-\tau)nR}\leq e^{(t- \tau) R}$, $\tau\in[0,t]$, $R<0$, we get
 \begin{align*}
 \|k_t^{(n)}\|_n&\leq e^{tnR} C^n n!
+n!(n-1)\int_0^t e^{(t-\tau)nR}e^{\tau R} (C+\tau)^{n-1} d\tau\\
&  \leq e^{tR} C^n n! +
n!(n-1)\int_0^t e^{(t-\tau)R}e^{\tau R} (C+\tau)^{n-1} d\tau\\
& =e^{tR} C^n n! +
n!(n-1)e^{t R}\int_0^t (C+\tau)^{n-1} d\tau\\
& = e^{tR} C^n n! +
n! (n-1) e^{t R}  \frac{(C+t)^{n}-C^n}{n}\leq (C+t)^{n}n!e^{t R} .
\end{align*}

 Let now $R\geq0$. Then, assuming that 
 \[
 \|k_\tau^{(n-1)}\|_{n-1}\leq e^{\tau (n-1)R} (C+\tau)^{n-1}(n-1)!, \qquad \tau\geq0, 
 \]
 we get
\begin{align*}
 \|k_t^{(n)}\|_n&\leq e^{tnR} C^n n!
+n!(n-1)B^{n-1}\int_0^t e^{(t-\tau)nR}e^{\tau (n-1)R} (C+\tau)^{n-1} d\tau\\
&  = e^{tnR} C^n n!
+e^{t nR} n!(n-1)\int_0^t e^{-\tau R}  (C+\tau)^{n-1} d\tau \\
&  \leq e^{tnR} C^n n!
+e^{t nR} n!(n-1)\frac{(C+t)^{n}-C^n}{n}\leq (C+t)^{n}e^{tnR}
n!.
\end{align*}

3) We rewrite \eqref{eq:repr} in the form
\begin{equation}\label{eq:repr1}
\begin{aligned}
 k_t^{(n)}(s_1,x_1,\dotsc,s_n,x_n)&=e^{t(A^{(n)}-B^{(n)})}k_0^{(n)}(s_1,x_1,\dotsc,s_n,x_n)
 \\
 &\quad+\int_0^t e^{(t-\tau)(A^{(n)}-B^{(n)})}(W^{(n)}k_\tau^{(n-1)})(s_1,x_1,\dotsc,s_n,x_n)d\tau.
\end{aligned}
\end{equation}
By the Trotter--Lie formula,
\[
  e^{t(A^{(n)}-B^{(n)})}=\lim_{m\to\infty}\bigl(e^{\frac{t}{m}A^{(n)}}e^{-\frac{t}{m}B^{(n)}}\bigr)^m.
\]
For any $f_n\in\CX_n^+$, $n\in\LN$, $\hat{x}^{(n)}:=(s_1,x_1,\ldots,s_n,x_n)$, we get, using the notation \eqref{ki},
\begin{align*}
(A^{(n)}f_n)(\hat{x}^{(n)}) &= \sum_{i=1}^n
  \int_{\Y}q(s,s_i)a(x-x_i)f_n(s,x)\nu (ds)\sigma(dx)\\
&\geq \alpha \sum_{i=1}^n
  \int_{B}f_n(s,x)\nu (ds)\sigma(dx).
\end{align*}
Therefore, if $b_n>0$ is such that
\begin{equation}\label{eqLaafqr}
  f_n(\hat{x}^{(n)})\geq b_n, \qquad \hat{x}^{(n)}\in B^n,
\end{equation}
then
\[
  (A^{(n)}f_n)(\hat{x}^{(n)})\geq n b_n \beta,   \qquad \hat{x}^{(n)}\in B^n,
\]
where $\beta:=\alpha (\nu\otimes\sigma)(B)$. Iterating, one gets for each $j\in\LN$,
\[
  ((A^{(n)})^jf_n)(\hat{x}^{(n)})\geq n^j b_n \beta^j,   \qquad \hat{x}^{(n)}\in B^n,
\]
and hence, for any $\tau>0$,
\[
  (e^{\tau A^{(n)}}f_n)(\hat{x}^{(n)})\geq  b_n e^{n\beta \tau},   \qquad \hat{x}^{(n)}\in B^n.
\]

Let $\mu=\|m\|_{L^\infty(\LR_+^*,d\nu)}$. Then \eqref{eqLaafqr} implies
\[
  (e^{-\tau B^{(n)}} f_n)(\hat{x}^{(n)})\geq e^{-\mu n \tau} f_n(\hat{x}^{(n)})\geq e^{-\mu n \tau} b_n,   \qquad \hat{x}^{(n)}\in B^n.
\]
Therefore, 
\[
  (e^{\tau A^{(n)}}e^{-\tau B^{(n)}} f_n)(\hat{x}^{(n)})\geq
  e^{(\beta-\mu) n \tau} b_n, \qquad \hat{x}^{(n)}\in B^n,
\]
and hence
\[
  ((e^{\frac{t}{m} A^{(n)}}e^{-\frac{t}{m} B^{(n)}})^m f_n)(\hat{x}^{(n)})\geq
  e^{(\beta-\mu) n t} b_n, \qquad \hat{x}^{(n)}\in B^n.
\]

Consider $n=1$. Then, for any $(s,x)\in B$ and $t\geq0$,
\[
  k_t^{(1)}(s,x)= e^{t(A^{(1)}-B^{(1)})}k_0^{(1)}(s,x)\geq \alpha e^{(\beta- \mu)t}.
\]

Let now $n\geq2$. Suppose that, for all $\tau\geq T_{n-1}$,
\[
  k_\tau^{(n-1)}(\hat{x}^{(n-1)})\geq \alpha^{n-1}e^{(n-1)(\beta-\mu)\tau}(n-1)!, \qquad \hat{x}^{(n-1)}\in B^{n-1}.
\]
Then
\[
  (W^{(n)}k_\tau^{(n-1)})(\hat{x}^{(n)})\geq \alpha^n n(n-1) e^{(n-1)(\beta-\mu)\tau}(n-1)!, \qquad \hat{x}^{(n)}\in B^{n},
\]
and therefore, by \eqref{eq:repr1}, for $n\geq2$, $t\geq T_{n}$, and $\hat{x}^{(n)}\in B^n$,
\begin{align*}
 k_t^{(n)}(\hat{x}^{(n)})&\geq  \alpha^n \int_0^t 
e^{(\beta-\mu) n (t- \tau)} n(n-1) e^{(n-1)(\beta-\mu)\tau}(n-1)! d\tau
\\&\geq \alpha^n n!e^{(\beta-\mu) nt}  (n-1)
\int_{T_{n-1}}^t 
e^{-(\beta-\mu)\tau} d\tau
\geq \alpha^n n!e^{(\beta-\mu) nt} (n-1)(t -T_{n-1})\\&\geq
 \alpha^n n!e^{(\beta-\mu) nt} (n-1)(T_n -T_{n-1})=  \alpha n!e^{(\beta-\mu) nt}.
\end{align*}
The statement is fully proved. 
\end{proof}
 
\section{Bolker--Dieckmann--Law--Pacala model}
We modify the contact model by adding a competition term, see e.g. \cite{bolkpaca97,dieckmannlaw,MR3280962,MR2505861}. The model is given by the following operator for $F\in K(\bbs)$:
\begin{align*}
 (LF)(\eta)&=\sum_{x\in\tau(\eta)}m(s_x)[F(\eta-s_x\delta_x)-F(\eta)]
 \\
 &\quad+\sum_{x\in\tau(\eta)}\sum_{y\in\tau(\eta)\setminus\{x\}}q^-(s_x,s_y)a^-(x-y)\left[F(\eta-s_x\delta_x)-F(\eta)\right]
 \\
 &\quad+\sum_{x\in\tau(\eta)}\int_{\Y}q^+(s_x,s)a^+(x-y)[F(\eta+s\delta_y)-F(\eta)]\nu(ds)\sigma(dy).
\end{align*}
Here $m:\LR_+^*\to\LR_+$ is the mortality rate function, $0\leq a\in \pm\in L^1(\X,d \sigma)\cap L^\infty(\X,d \sigma)$ are spatial dispersion and competition kernels, such that $a^\pm(-x)=a^\pm(x)$, $x\in\X$; and $q^\pm :\LR_+^*\times\LR_+^*\to\LR_+$ are symmetric functions. We denote
\[
  \kappa^\pm:=\int_\X a^\pm(x)\sigma(dx)>0.
\]

\begin{Proposition}
For any $G\in B_\mathrm{bs}(\LK_0(\X))$, 
\[
  \hat{L}G:=K^{-1}LKG=\hat{L}_0G+\hat{L}_1G+\hat{L}_2G+\hat{L}_3G,
\] 
where
 \begin{align*}
  (\hat{L}_0G)(\eta)&:=- \sum_{x\in\tau(\eta)} m(s_x)G(\eta)- \sum_{x\in\tau(\eta)} \sum_{y\in\tau(\eta)\setminus\{x\}}q^-(s_x,s_y)a^-(x-y)G(\eta),
  \\
  (\hat{L}_1G)(\eta)&:=-\sum_{x\in\tau(\eta)}\sum_{y\in\tau(\eta)\setminus\{x\}}q^-(s_x,s_y)a^-(x-y)G(\eta-s_x\delta_x)
  \\
 (\hat{L}_2G)(\eta)&:=\sum_{x\in\tau(\eta)}\int_{\Y}q^+(s_x,s)a^+(x-y)G(\eta-s_x\delta_x+s\delta_y)\nu(ds)\sigma(dy)
  \\
  (\hat{L}_3G)(\eta)&:=\sum_{x\in\tau(\eta)}\int_{\Y}q^+(s_x,s)a^+(x-y)G(\eta+s\delta_y)\nu(ds)\sigma(dy).
 \end{align*}
\end{Proposition}

For $C>0$ and $\alpha\in\LR$, we set
\[
  \IC(\eta):=C^{|\tau(\eta)|}e^{\alpha\sum_{y\in\tau(\eta)}s_y}, \qquad \eta\in\LK_0(\X),
\]
and define the space 
\begin{equation}
\mathbf{L}_{\alpha,C}:=L^1\left(\LK_0(\X),\IC(\eta)d\lambda_\theta(\eta)\right). \label{LalphaC}
\end{equation}
We denote its norm by $\|\cdot\|_{\alpha,C}$.

We define
\[
  D(\eta):=\sum_{x\in\tau(\eta)}  m(s_x)+\sum_{x\in\tau(\eta)} \sum_{y\in\tau(\eta)\setminus\{x\}}q^-(s_x,s_y)a^-(x-y) , \qquad \eta\in\LK_0(\X),
\]
and consider also the linear set
\[
  \CD := \{G\in\mathbf{L}_{\alpha,C} : DG \in \mathbf{L}_{\alpha,C} \}.
\]

\begin{Theorem}
Let $C>0$ and $\alpha\in\LR$.
Suppose that there exist $\beta>0$, such that 
\begin{align}\label{betaminus}
&\int_{\LR_+^*}q^-(s,\tau)e^{\alpha \tau}\nu(d\tau)\leq \beta m(s), \qquad s>0,
\\\label{betaplus}
&\int_{\LR_+^*}q^+(s,\tau)e^{\alpha \tau}\nu(d\tau)\leq \beta e^{\alpha s} m(s), \qquad s>0,\\\label{betacompare}
&q^+(s,\tau)a^+(x) \leq \beta e^{\alpha \tau} q^-(s,\tau) a^-(x), \qquad s,\tau>0, \ x\in\X,\\
\label{eq:smallbeta}
&\kappa^++\kappa^- C+\frac{1}{C}<\frac{1}{2\beta}.
\end{align}
Then $(\hat{L},\CD)$ is the generator of an analytic semigroup $T(t)$, $t\geq0$, in $\mathbf{L}_{\alpha,C}$.
\end{Theorem}

\begin{proof}
Firstly, using the same arguments as in \cite[Lemma~3.3]{MR2863863}, we can show that $(\hat{L}_0,\CD)$ is the generator of an analytic contraction semigroup in $\mathbf{L}_{\alpha,C}$. 

Next, we recall (see e.g. \cite{MR1721989}) that, for a Banach space $Z$, a linear operator $(B,{D}(B))$ is  (relatively) $A$-bounded w.r.t. a linear operator $(A,{D}(A))$, if ${D}(A)\subseteq{D}(B)$ and if there exist constants $a,b\in\LR_+$ such that
\begin{equation}\label{relbd}
 \|Bx\|\leq a\|Ax\|+b\|x\|
\end{equation}
for all $x{D}(A)$. The $A$-bound of $B$ is
\begin{equation*}
 a_0:=\inf\{a\geq 0: \exists b\in\LR_+\text{ such that \eqref{relbd} holds}\}
\end{equation*}
For $A$ being the generator of an analytic semigroup,  $(A+B,{D}(A))$ generates an analytic semigroup for every $A$-bounded operator $B$ having $A$-bound $a_0<\frac12$.

We are going to show now that, under assumptions above, the operator $\hat{L}_1+\hat{L}_2+\hat{L}_3$ is $\hat{L}_0$-bounded. 
Indeed, for each $G\in\CD$,
  \begin{align*}
   \|\hat{L}_1G\|_{\alpha,C}
   &\leq\int_{\LK_0(\X)} \sum_{x\in\tau(\eta)}\sum_{y\in\tau(\eta-s_x\delta_x)}q^-(s_x,s_y)a^-(x-y)|G(\eta-s_y\delta_y)|\IC(\eta)\lambda(d\eta)
   \\
   &\stackrel{\mathmakebox[\widthof{=}]{\eqref{Minlos}}}= \int_{\LK_0(\X)}\int_{\Y}\sum_{x\in\tau(\eta)}q^-(s_x,s)a^-(x-y)|G(\eta)|\IC(\eta+s\delta_y)\nu(ds)\sigma(dy)\lambda(d\eta)
   \\
   &= \kappa^-C\int_{\LK_0(\X)}|G(\eta)|\IC(\eta)\sum_{x\in\tau(\eta)}\int_{\LR_+^*}q^-(s_x,s)e^{\alpha s} \nu(ds) \lambda(d\eta)
   \\
   &\leq\beta\kappa^-C\int_{\LK_0(\X)}\sum_{x\in\tau(\eta)}m(s_x) |G(\eta)|\IC(\eta)\lambda(d\eta)
  \leq \beta\kappa^-C \|\hat{L}_0 G\|_{\alpha,C},
  \end{align*}
  where we used \eqref{betaminus}.
Next, by \eqref{betaplus}, we have, for each $G\in\CD$,
  \begin{align*}
   \|&\hat{L}_2G\|_{\alpha,C}
   \\
   &\leq \int_{\LK_0(\X)}\sum_{x\in\tau(\eta)}\int_{\Y}q^+(s_x,s)a^+(x-y)|G(\eta-s_x\delta_x+s\delta_y)|\nu(ds)\sigma(dy)\IC(\eta)\lambda(d\eta)
   \\
   &\stackrel{\mathmakebox[\widthof{=}]{\eqref{Minlos}}}=\kappa^+\int_{\LK_0(\X)}\sum_{y\in\tau(\eta)}\int_{\Y}q^+(s,s_y)a^+(x-y)|G(\eta)|\IC(\eta-s_y\delta_y+s\delta_x)\nu(ds)\sigma(dx)\lambda(d\eta)
   \\
   &= \int_{\LK_0(\X)}|G(\eta)|\IC(\eta)\sum_{y\in\tau(\eta)}\int_{\Y}q^+(s,s_y)e^{-\alpha s_y}e^{\alpha s }a^+(x-y)\nu(ds)\sigma(dx)\lambda(d\eta)\\
   &\leq \kappa^+ \beta \int_{\LK_0(\X)}|G(\eta)|\IC(\eta)\sum_{y\in\tau(\eta)} e^{-\alpha s_y}e^{\alpha s_y }m(s_y)\lambda(d\eta)
  \\& \leq \kappa^+ \beta\|\hat{L}_0 G\|_{\alpha,C}.
  \end{align*} 
Finally, by \eqref{betacompare}, we get, for any $G\in\CD$, 
  \begin{align*}
   \|&\hat{L}_3G\|_{\alpha,C} 
   \\
   &\leq\int_{\LK_0(\X)}\sum_{x\in\tau(\eta)}\int_{\Y}q^+(s_x,s)a^+(x-y)|G(\eta+s_y\delta_y)|\IC(\eta)\nu(ds)\sigma(dy)\lambda(d\eta)
   \\
   &\stackrel{\mathmakebox[\widthof{=}]{\eqref{Minlos}}}=\int_{\LK_0(\X)}\sum_{y\in\tau(\eta)}\sum_{x\in\tau(\eta)\setminus\{y\}} q^+(s_x,s_y)a^+(x-y)|G(\eta)|\IC(\eta-s_y\delta_y)\lambda(d\eta)
   \\
   &\leq \beta\int_{\LK_0(\X)}\sum_{y\in\tau(\eta)}\sum_{x\in\tau(\eta)\setminus\{y\}}
    q^-(s_x,s_y)a^-(x-y)e^{\alpha s_y} C^{-1}e^{-\alpha s_y} \IC(\eta) \lambda(d\eta)
   \\
   &\leq\frac{\beta}{C} \|\hat{L}_0G\|_{\alpha,C}.
  \end{align*}
Combining the estimates with the assumption \eqref{eq:smallbeta}, one gets the statement.
\end{proof}

\begin{Proposition} For any $k\in\fin$, the mapping 
 \begin{align*}
  (L^\triangle k)(\eta)&=-D(\eta)k(\eta)
  \\
  &\quad-\int_{\Y}\sum_{x\in\tau(\eta)}q^-(s_x,s)a^-(x-y)k(\eta+s\delta_y)\nu(ds)\sigma(dy)
  \\
  &\quad+\sum_{y\in\tau(\eta)}\int_{\Y}q^+(s,s_y)a^+(x-y)k(\eta-s_y\delta_y+s\delta_x)\nu(ds)\sigma(dx)
  \\
  &\quad+\sum_{y\in\tau(\eta)}\sum_{x\in\tau(\eta)\setminus\{y\}}q^+(s_x,s_y)a^+(x-y)k(\eta-s_y\delta_y)
 \end{align*}
is well-defined and, for any $G\in\bbs$, \eqref{eq:duality} holds. 
\end{Proposition}

Hence, one can consider the dual semigroup $T^*(t)$ in the dual space to the space $\mathbf{L}_{\alpha,C}$, which is isomorphic to
\begin{equation}\label{KalphaC}
 \mathcal{K}_{\alpha,C}=L^\infty\bigl(\LK_0(\X),\mathbf{C}(\eta)\lambda_{\nu\otimes\sigma}(d\eta)\bigr).
\end{equation}
This semigroup is $*$-weakly differentiable (with respect to the duality \eqref{eq:duality}), and $k_t=T^*(t)k_0$ solves the weak equation \eqref{eqDualc} (see \cite{MR2863863} for details). Note that, for some $A,B>0$, 
\[
  \|T^*(t)\|=\|T(t)\|\leq A e^{Bt}, \qquad t>0.
\] 
Therefore, for $\lambda_{\nu\otimes\sigma}$-a.a. $\eta\in\LK_0(\X)$,
\begin{equation}\label{subPoisest}
 |k_t(\eta)|\leq A e^{Bt} C^{|\tau(\eta)|}e^{\alpha\sum_{x\in\tau(\eta)}s_x}.
\end{equation}
As we can see, comparing with the result \eqref{apriori1}, the strong mortality and competition rates prevent factorial growth of correlation functions in $n$. A further analysis of the classical solution to the strong equation \eqref{qfpe} can be done by using the sun-dual semigroup techniques, see \cite{MR2863863} for details.

\section{Glauber dynamics}

We consider now the Glauber-type dynamics. The corresponding analogue on the configuration spaces was studied in many papers, see e.g. \cite{MR1899231,MR2144229,MR2863863,MR2417815}. The generator of the Glauber dynamics is obtained from the Gibbs measure on the cone, which was constructed in \cite{MR3041709} as follows.  Let $\X=\LR^d$ and consider a pair potential
\begin{equation*}
 \phi:\X\times\X\to\LR
\end{equation*}
which satisfies the following two properties:
\begin{itemize}
 \item there exists $R>0$ such that
 \begin{equation*}
  \phi(x,y)=0\text{ if }|x-y|>R
 \end{equation*}
 (where $|\cdot|$ denotes the Euclidean norm on $\LR^d$);
 \item there exists $\delta>0$ such that
 \begin{equation*}
  \inf_{|x-y|\leq\delta}\phi(x,y)>2b_d c^d\sup_{x,y}\bigl\vert-\phi(x,y)\lor 0\bigr\vert,
 \end{equation*}
 where $b_d$ is the volume of a unit ball in $\LR^d$ and 
 $c=c_{d,\delta,R} := \sqrt{d}(1+R/\delta)$ (see \cite{MR3041709} for details).
\end{itemize}

Fix also a $\theta>0$. It was shown in \cite{MR3041709} that there exists a \term{tempered Gibbs measure} $\mu$ on $\LK(\LR^d)$ which fulfills the Dobrushin-Lanford-Ruelle equations
\begin{equation*}
 \int_{\LK(\LR^d)}\pi_\Delta(B\mid\eta)\mu(d\eta)=\mu(B)\text{ for any }\Delta\in\CB_c(\LR^d),
\end{equation*}
where $\pi_\Delta$ is the so-called \term{local specification} constructed by $\phi$ and $\theta$ (see \cite{MR3041709} for the precise definitions and further details). Heuristically, 
\[
\mu(d\eta) = \frac{1}{Z} \exp \Bigl( - \sum_{x,y\in\tau(\eta)}s_x s_y \phi(x,y) \Bigr)\CG_\theta(d\eta),
\]
where $Z$ is a normalizing factor. 
\begin{Proposition}[Georgii--Nguyen--Zessin identity, {\cite[Theorem 6.4]{MR3041709}}]
Let $\mu$ be a tempered Gibbs measure on $\LK(\X)$. Then, for any measurable function $F:\LR^d\times\LK(\LR^d)\to\LR_{+}$,
 \begin{align}
  \begin{split}\label{gnz}
   \int_{\LK(\LR^d)}&\int_{\LR^d}F(x,\eta)\eta(dx)\mu(d\eta)
   \\
   &=\int_{\LK(\LR^d)}\int_{\LR_{+}^*\times\LR^d}F(x,\eta+s\delta_x)e^{-\Phi((s,x);\eta)}s \nu_{\theta}(ds)\sigma(dx)\mu(d\eta),
  \end{split}
 \end{align}
 where, for $\eta:=(s_y,y)_{y\in\tau(\eta)}\in\LK(\LR^d)$,
 \begin{equation}\label{defPhi}
  \Phi\left((s,x);\eta\right):=2s\sum_{y\in\tau(\eta)}s_y\phi(x,y).
 \end{equation}
\end{Proposition}

We consider, for the fixed $\phi$ and $\theta$, the following (pre-)Dirichlet form:
\begin{equation*}
 \mathcal{E}(F,G):= \int_{\LK(\X)}\int_{\X}D_x^-F(\eta)D_x^-G(\eta)\eta(dx)\CG_\theta(d\eta),
\end{equation*}
for $F,G \in K(\bbs)$.

\begin{Proposition}
Let $F,G\in K(\bbs)$ and
 \begin{align*}
  (LF)(\eta)&:=\sum_{x\in\tau(\eta)}s_x\left[F(\eta-s_x\delta_x)-F(\eta)\right]
  \\
  &\quad+\int_{\Y}\left[F(\eta+s_x\delta_x)-F(\eta)\right]e^{-\Phi\left((s,x);\eta)\right)}s \nu_\theta(ds)\sigma(dx)
 \end{align*}
 where $\Phi$ is defined by \eqref{defPhi}. Then
 \begin{equation*}
 \mathcal{E}(F,G)=-\int_{\LK(\X)} (LF)(\eta) \, G(\eta)\CG_\theta(d\eta).
 \end{equation*}
 
\end{Proposition}
\begin{proof}
By \eqref{gnz}, we have
 \begin{align*}
  \mathcal{E}(F,G)&=\int_{\LK(\X)}\int_{\Y}D_x^-F(\eta)D_x^-G(\eta)\eta(dx)\CG_\theta(d\eta)
  \\
  &=\frac{1}{2}\int_{\LK(\X)}\int_{\Y}D_x^-F(\eta)(G(\eta-s_x\delta_x)-G(\eta))\eta(dx)\CG_\theta(d\eta)
  \\
  &=\int_{\LK(\X)}\int_{\Y}D_x^-F(\eta)(G(\eta-s_x\delta_x)\eta(dx)\CG_\theta(d\eta)
  \\
  &\quad-\int_{\LK(\X)}\int_{\Y}D_x^-F(\eta)G(\eta)\eta(dx)\CG_\theta(d\eta)
  \\
  &\stackrel{\mathmakebox[\widthof{=}]{\eqref{gnz}}}=\frac{1}{2}\int_{\LK(\X)}\int_{\Y}D_x^-F(\eta+s_x\delta_x)G(\eta)e^{-\Phi\left((s,x);\eta\right)}s\nu_\theta(ds)\sigma(dx)\CG_\theta(d\eta)
  \\
  &\quad-\int_{\LK(\X)}\int_{\Y}D_x^-F(\eta)G(\eta)\eta(dx)\CG_\theta(d\eta)
  \\
  &=-\int_{\LK(\X)}\int_{\Y}(F(\eta+s_x\delta_x)-F(\eta))G(\eta)e^{-\Phi\left((s,x);\eta\right)}s\nu_\theta(ds)\sigma(dx)\CG_\theta(d\eta)
  \\
  &\quad-\int_{\LK(\X)}\sum_{x\in\tau(\eta)}s_x(F(\eta-s_x\delta_x)-F(\eta))G(\eta)\CG_\theta(d\eta)
  \\
  &=-\int_{\LK(\X)} (LF)(\eta) \, G(\eta)\CG_\theta(d\eta).\qedhere
 \end{align*}
\end{proof}

We denote
\[
  S(\eta):= \sum_{x\in\tau(\eta)}s_x, \qquad \eta\in\LK_0(\X),
\]
and also, for the fixed $(s,x)\in\Y$, we set
\[
  f_{s,x}(\tau,y):=e^{-2s\tau\phi(x,y)}-1, \qquad (\tau,y)\in\Y.
\]

\begin{Proposition}
For any $G\in B_\mathrm{bs}(\LK_0(\X))$, $\hat{L}G:=K^{-1}LKG$ satisfies 
\begin{align*}
 (\hat{L}G)(\eta)&=-S(\eta)G(\eta)
 \\
 &\hspace{11pt}+\int_{\Y}s\sum_{\xi\subset\eta}G(\xi+s\delta_x)e^{-\Phi((s,x),\xi)}e_\lambda(f_{s,x},\eta-\xi)\nu_\theta(ds)\sigma(dx).
\end{align*}
 \end{Proposition}
\begin{proof}
 The proof can be done in the same way as e.g. in \cite{MR2863863}, see also the proof of \eqref{hatL}.
\end{proof}

We consider again the space \eqref{LalphaC} with $C>0$, $\alpha\in(0,1)$. We consider also 
the domain
\begin{equation*}
  \mathcal{D}:=\left\{G\in\mathbf{L}_{\alpha,C}\middle| S(\eta) G(\eta)\in\mathbf{L}_{\alpha,C}\right\}.
 \end{equation*}

\begin{Theorem}
Let $C> 2$, $\alpha\in(0,1)$, $\theta>0$, and
\begin{equation}\label{pos}
 \phi(x,y)\geq0, \qquad x,y\in\X,
 \end{equation}
be such that
\begin{equation}\label{smallparam}
\theta\cdot\sup_{x\in\X}\int_\X \phi(x,y)\sigma(dy) \leq \frac{\alpha(1- \alpha)}{2C}.
\end{equation}
Then the operator $(\hat{L},\mathcal{D})$ generates an analytic semigroup in the space $\mathbf{L}_{\alpha,C}$.
\end{Theorem}
\begin{proof}
Firstly, we consider the operator
\[
  (\hat{L}_0G)(\eta)=-S(\eta)G(\eta), \qquad \eta\in\LK_0(\X)
\]
in $\mathbf{L}_{\alpha,C}$ with its maximal domain $\mathcal{D}$. It can be shown identically to the proof of \cite[Lemma~3.3]{MR2863863}, that  $\bigl(\hat{L}_0,\mathcal{D}\bigr)$ is a generator of a contraction analytic semigroup in $\mathbf{L}_{\alpha,C}$.

Next, we define $\hat{L}_1:=\hat{L}-\hat{L}_0$, i.e., for $G\in\bbs$,
\begin{align*}
 (\hat{L}_1G)(\eta)&=\int_{\Y}s\sum_{\xi\subset\eta}G(\xi+s\delta_x)e^{-\Phi((s,x),\xi)} e_\lambda(f_{s,x},\eta-\xi)\nu_\theta(ds)\sigma(dx).
\end{align*}

We are going to show now that, under \eqref{pos}, operator $\hat{L}_1$ is $\hat{L}_0$-bounded. Indeed, 
 \begin{align*}
  &\|\hat{L}_1G\|_{\alpha,C}
  \\
  &
  \stackrel{\mathmakebox[\widthof{=}]{{\eqref{pos}}}}\leq
\int_{\LK_0(\X)}\sum_{\xi\subset\eta}\int_{\Y}s |G(\xi+s\delta_x)|  e_\lambda\left(\left|f_{s,x}\right|,\eta-\xi\right)\IC(\eta)\nu_\theta(ds)\sigma(dx)\lambda_\theta(d\eta)
  \\
  &\stackrel{\mathmakebox[\widthof{=}]{{\eqref{Minlos}}}}=\int_{\LK_0(\X)}\int_{\LK_0(\X)}\int_{\Y}s |G(\xi_1+s\delta_x)|  e_\lambda\left(\left|f_{s,x}\right|,\xi_2\right)\IC(\xi_1+\xi_2)\nu_\theta(ds)\sigma(dx)\lambda_\theta(d\xi_1)\lambda_\theta(d\xi_2)
  \\
  &\stackrel{\mathmakebox[\widthof{=}]{{\eqref{Minlos}}}}=\int_{\LK_0(\X)}\int_{\LK_0(\X)}\sum_{x\in\tau(\xi_1)}s_x |G(\xi_1)| e_\lambda\left(\left|f_{s_x,x}\right|,\xi_2\right)\IC(\xi_1-s_x\delta_x+\xi_2)\lambda_\theta(d\xi_1)\lambda_\theta(d\xi_2)
  \\
  &\leq\int_{\LK_0(\X)}|G(\xi_1)|\IC(\xi_1)\sum_{x\in\tau(\xi_1)}\int_{\LK_0(\X)}s_x e_\lambda\left(\left|f_{s_x,x}\right|,\xi_2\right)\IC(\xi_2)\IC(s_x\delta_x)^{-1}\lambda_\theta(d\xi_1)\lambda_\theta(d\xi_2)
  \\
  &=C^{-1}\int_{\LK_0(\X)}|G(\xi_1)|\IC(\xi_1)\sum_{x\in\tau(\xi_1)}s_xe^{-\alpha s_x} \int_{\LK_0(\X)}e_\lambda\left(\left|f_{s_x,x}\right|Ce^{\alpha s_\cdot},\xi_2\right)\lambda_\theta(d\xi_2)\lambda_\theta(d\xi_1)
  \\
  &=C^{-1}\int_{\LK_0(\X)}|G(\xi_1)|\IC(\xi_1)\sum_{x\in\tau(\xi_1)}s_xe^{-\alpha s_x} \exp\left(\int_{\Y}\left|f_{s_x,x}(s,y)\right|Ce^{\alpha s}\nu_\theta(ds)\sigma(dy)\right)\lambda_\theta(d\xi_1)
  \\
  \intertext{where we have used \eqref{eq:intexp}; next, since, under \eqref{pos}, $|f_{s,x}(\tau,y)|\leq 2s\tau\phi(x,y)$, we may estimate}
  &\leq C^{-1}\int_{\LK_0(\X)}|G(\xi_1)|\IC(\xi_1)\sum_{x\in\tau(\xi_1)}s_xe^{-\alpha s_x} \exp\left(\int_{\Y} 2s_xs\phi(x,y)Ce^{\alpha s}\nu_\theta(ds)\sigma(dx)\right)\lambda_\theta(d\xi_1)
  \\
  &\stackrel{\mathmakebox[\widthof{=}]{\eqref{gamma}}}=C^{-1}\int_{\LK_0(\X)}|G(\xi_1)|\sum_{x\in\tau(\xi_1)}s_x \exp\left[\left(2C\theta \int_{\Y}\phi(x,y)e^{(\alpha-1) s}ds\sigma(dx)-\alpha\right)s_x\right] \IC(\xi_1)\lambda_\theta(d\xi_1)
  \\
  &\leq\frac{1}{C}\|\hat{L}_0 G\|,
 \end{align*}
for $\alpha\in(0,1)$, if only we assume that
 \begin{align*}
  0&\geq2C\theta\int_{\Y}\phi(x,y)e^{(\alpha-1)s}ds\sigma(dx)-\alpha
  =2C\theta\int_{\X}\phi(x,y)\sigma(dx)\int_{\LR_+^*}e^{(\alpha-1)s}ds-\alpha
  \\
  &=\frac{2C}{1-\alpha}\theta\int_{\LR^d}\phi(x,y)\sigma(dx)-\alpha,
 \end{align*}
 that holds under \eqref{smallparam}. 
 Therefore, $\hat{L}_1$ has $\hat{L}_0$-bound $\frac{1}{C}<\frac12$ that yields the statement. 
\end{proof}

By using the Minlos identity \eqref{Minlos}, we immediately get the following result:
\begin{Proposition}
 For any $k\in\fin$, the mapping
 \begin{align*}
  \left(L^\triangle k\right)(\eta)
  &=-S(\eta) k(\eta)\\&\quad+
\sum_{x\in\tau(\eta)}s_x e^{-\Phi((s,x),\eta-s_x\delta_x)}
  \int_{\LK_0(\X)}   e_\lambda(f_{s_x,x},\xi)k(\eta+\xi-s_x\delta_x)\lambda_\theta(d\xi)
 \end{align*}
is well-defined and, for any $G\in\bbs$, \eqref{eq:duality} holds.
\end{Proposition}

Again, one can consider the dual semigroup $T^*(t)$ in the space (isomorphic to) $\mathcal{K}_{\alpha,C}$ given by \eqref{KalphaC}, so that $k_t=T^*(t)k_0$ solves the weak equation \eqref{eqDualc}, and \eqref{subPoisest} holds.  Further analysis of the sun-dual semigroup $T^\odot(t)$ (which provides a solution to \eqref{eqDualc} on a subspace of $\mathcal{K}_{\alpha,C}$) can be done in the same way as in \cite{MR2863863}. 

\section*{Acknowledgements}
P.K. was supported by the DFG through the IRTG 2235 ``Searching for the regular in the irregular: Analysis of singular and random systems''.

\end{document}